\documentclass{article}
\usepackage[margin=3cm]{geometry}
\usepackage{graphicx,color}
\usepackage{amssymb,amsthm}
\usepackage{caption} 
\graphicspath{{Figures/}}
\usepackage{wrapfig}
\usepackage{url}
\theoremstyle{plain}
\newtheorem{theorem}{Theorem}
\newtheorem{lemma}{Lemma}

\newtheorem{proposition}{Proposition}
\theoremstyle{definition}
\newtheorem{definition}{Definition}

\newtheoremstyle{nodot}
  {}
  {}
  {\it}
  {}
  {\bf}
  {}
  { }
  {}%
\theoremstyle{nodot}

\newcommand{\ITEMMACRO}[2]{\par\noindent%
\hangindent=#2em\setbox0\hbox{#1 \kern5pt}\ifdim\wd0<\hangindent\setbox0\hbox to\hangindent{\hss#1\quad}\fi\box0\ignorespaces}

\newcommand{\Item}[1]{\ITEMMACRO{#1}{2.9}}
\newcommand{\Bitem}{\ITEMMACRO{$\bullet$}{2.9}}

\newcommand{\RR}{\mathbb{R}}
\newcommand{\FF}{{\cal F}}
\newcommand{\qedclaim}{\hfill$\triangle$\smallskip}
\begin{document}
\setlength\parindent{0pt} 
\advance\parskip4pt

\title{Straight Line Triangle Representations}



\author{Nieke Aerts         \and
        Stefan Felsner 
}
\date{}



\maketitle

\begin{abstract}
  A straight line triangle representation (SLTR) of a planar graph is a
  straight line drawing such that all the faces including the outer face have
  triangular shape. Such a drawing can be viewed as a tiling of a triangle
  using triangles with the input graph as skeletal structure. In this paper we
  present a characterization of graphs that have an SLTR. The characterization
  is based on flat angle assignments, i.e., selections of angles of the graph
  that have size~$\pi$ in the representation.  We also provide a second
  characterization in terms of contact systems of pseudosegments. With the aid
  of discrete harmonic functions we show that contact systems of
  pseudosegments that respect certain conditions are stretchable. The
  stretching procedure is then used to get straight line triangle
  representations.  Since the discrete harmonic function approach is quite
  flexible it allows further applications, we mention some of them.

  The drawback of the characterization of SLTRs is that we are not able to
  effectively check whether a given graph admits a flat angle
  assignment that fulfills the conditions. Hence it is still open to
  decide whether the recognition of graphs that admit straight line
  triangle representation is polynomially tractable.
\end{abstract}

\section{Introduction}\label{intro}

In this paper we study a representation of planar graphs in the classical
setting, i.e., vertices are represented by points in the Euclidean plane and
edges are represented by non-crossing continuous curves connecting the points. We aim to
classify the class of planar graphs that admit a straight line
representation in which all faces are triangles. Haas et al.~present a
necessary and sufficient condition for a graph to be a
pseudo-triangulation~\cite{haas03}, however, this condition is not sufficient
for a graph to have a straight line triangle representation (e.g.~see
Figure~\ref{f_bad-faa} and~\cite{felaer12-a}).  There have been investigations
of the problem in the dual setting, i.e., in the setting of side contact
representations of planar graphs with triangles. Gansner, Hu and Kobourov show
that outerplanar graphs, grid graphs and hexagonal grid graphs are Touching
Triangle Graphs (TTGs). They give a linear time algorithm to find the
TTG~\cite{Gansner}. Alam, Fowler and Kobourov~\cite{alam12} consider proper
TTGs, i.e., the union of all triangles of the TTG is a triangle and there are
no holes. They give a necessary and a stronger sufficient condition for
biconnected outerplanar graphs to be TTG, a characterization, however, is
missing. Fowler has given a necessary and sufficient condition for a special
type of outerplanar graphs to be TTG~\cite{fowler13}. Kobourov, Mondal and
Nishat present construction algorithms for proper TTGs of 3-connected cubic
graphs and some grid graphs. They also present a decision algorithm for
testing whether a 3-connected planar graph is proper
TTG~\cite{kmn-ttr3cpg-12}. Gon\c{c}alves, L{\'e}v{\^e}que and Pinlou consider
a primal-dual contact representation by triangles, i.e., both the faces as
well as the vertices are represented by triangles. They show that all
3-connected planar graphs admit such a representation~\cite{GoncalvesLP12}.

Here is the formal introduction of the main character for this paper.

\begin{definition}[Straight Line Triangle Representation]\label{d_sltr} 
A plane drawing of a graph such that
\noindent \Item{-} all the edges are straight line segments and
\noindent \Item{-} all the faces, including the outer face, bound a
non-degenerate triangle

is called a \emph{Straight Line Triangle Representation} (SLTR).
\end{definition}

\bigskip

\noindent\begin{minipage}[t]{0.5\textwidth} \centering
\includegraphics[width=0.9\textwidth]{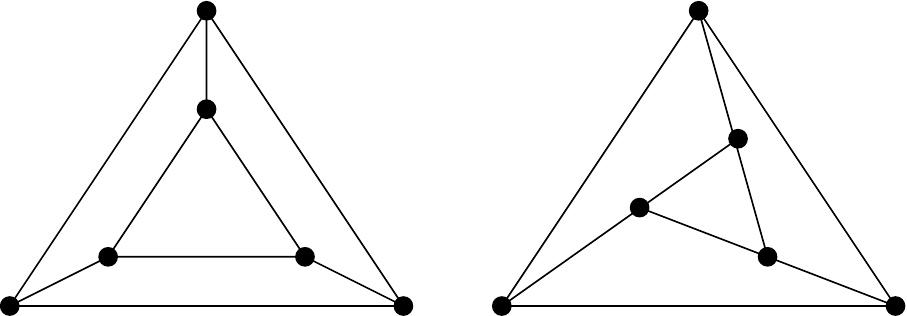}
\captionof{figure}{\small A graph and one of its
SLTRs}\label{f_exsltr}
\end{minipage}%
\begin{minipage}[t]{0.5\textwidth} \centering
\includegraphics[width=0.6\textwidth]{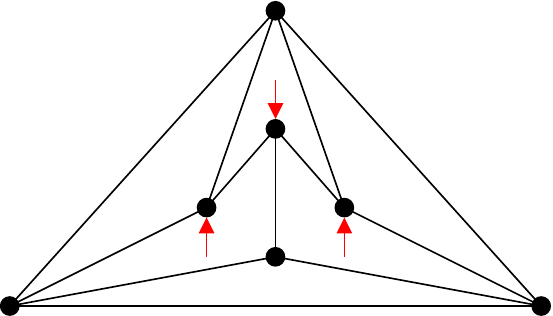}
\captionof{figure}{\small A Flat Angle Assignment arrows) that 
           has no corresponding SLTR.}\label{f_bad-faa}
\end{minipage} 

\bigskip

Clearly every straight line drawing of a triangulation is an SLTR.  So
the class of planar graphs admitting an SLTR is rich. On the other
hand, graphs admitting an SLTR cannot have a cut vertex.  Indeed, as
shown below (Proposition~\ref{p_int3}), graphs admitting an SLTR are well
connected.  Being well connected, however, is not sufficient as shown
e.g.\ by the cube graph.

To simplify the discussion we assume that the input graph is given
with a plane embedding and a selection of three vertices of the outer
face that are designated as corner vertices for the outer face.  These
three vertices are called \emph{suspension vertices}.  If needed, an
algorithm may try all triples of vertices as suspensions.

Every degree two vertex that is not a suspension is flat in every SLTR, i.e.,
it has angles of size~$\pi$ in both incident faces. Such a vertex and its two
incident edges can be replaced by a single edge connecting the two neighbors
of the vertex.  Such an operation is called a \textit{vertex reduction}.  
We use vertex reductions to eliminate all the
degree two vertices that are not suspensions.

A plane graph~$G$ with suspensions~$s_1,s_2,s_3$ is said to be
\textit{internally 3-connected} when the addition of a new
vertex~$v_\infty$ in the outer face, that is made adjacent to the
three suspension vertices, yields a 3-connected graph.

\begin{proposition}\label{p_int3} If a graph $G$ admits an SLTR with
$s_1,s_2,s_3$ as corners of the outer triangle and no vertex reduction
is possible, then $G$ is internally 3-connected.
\end{proposition}
\begin{proof} Consider an SLTR of $G$.  Suppose there is a separating
set $U$ of size~2. It is enough to show that each component of
$G\setminus U$ contains a suspension vertex, so that~$G + v_\infty$ is
not disconnected by $U$.  Since $G$ admits no vertex reduction every
degree two vertex is a suspension. Hence, if $C$ is a component and
$C\cup U$ induces a path, then there is a suspension in $C$. Otherwise
consider the convex hull of~$C\cup U$ in the SLTR. The convex corners
of this hull are vertices that expose an angle of size at least
$\pi$. Two of these large angles may be at vertices of~$U$ but there
is at least one additional large angle. This large angle must be the
outer angle at a vertex that is an outer corner of the SLTR, i.e., a
suspension.  
\end{proof}

From Proposition~\ref{p_int3}, it follows that any graph that is not
internally 3-connected but does admit an SLTR, is a subdivision of an
internally 3-connected graph. Therefore, we may assume that the graphs
we consider are internally 3-connected.

In Section~\ref{SLTR} we present necessary conditions for the
existence of an SLTR in terms of what we call a flat angle
assignment. A flat angle assignment that fulfills the conditions is
shown to induce a partition of the set of edges into a set of
pseudosegments. Finally, with the aid of discrete harmonic functions
we show that in our case the set of pseudosegments is stretchable.
Hence, the necessary conditions are also sufficient. The drawback of
the characterization is that we are not aware of an effective way of
checking whether a given graph admits a flat angle assignment that
fulfills the conditions.

Recently we have been able to give a second characterization of graphs that admit an SLTR using flat angle assignments and Schnyder labelings~\cite{felaer14}. Using this characterization it is easy to show that for graphs that have a unique Schnyder labeling (these graphs are identified by Felsner and Zickfeld\cite{FelsnerZ08}), the problem of deciding whether the graph has an SLTR can be translated into a matching problem in a bipartite graph. For graphs with very few Schnyder woods the problem also becomes polynomially tractable. However, there are planar 3-connected graphs on $n$ vertices which have $3.209^n$ Schnyder woods~\cite{FelsnerZ08}.

In Section~\ref{applications} we consider further applications of the
stretching approach. First we look at flat angle assignments that
yield faces with more than three corners. Then we proceed to prove a
more general result about stretchable systems of pseudosegments with
our technique. The result is not new, de Fraysseix and Ossona de
Mendez have investigated stretchability conditions for systems of
pseudosegments. The counterpart to Theorem~\ref{t_dFdM1} can be found
in~\cite[Theorem~38]{dfdm05}. The proof there is based on a long and
complicated inductive construction. 
The last section of the paper is dedicated to primal-dual contact representations by triangles. We give a simple proof of a theorem of Gon\c{c}alves, L{\'e}v{\^e}que and Pinlou, which shows that every 3-connected planar graph has a primal-dual contact representation by triangles~\cite{GoncalvesLP12}.

\section{Necessary and Sufficient Conditions}\label{SLTR}
\def\int{{\sf int}}

Consider a plane, internally 3-connected graph $G=(V,E)$ with suspensions given.
Suppose that $G$ admits an SLTR. This representation induces a set of
{\it flat angles}, i.e., incident pairs $(v,f)$ such that vertex $v$
has an angle of size $\pi$ in the face $f$.

Since $G$ is internally 3-connected every vertex has at most one flat
angle. Therefore, the flat angles can be viewed as a partial mapping
of vertices to faces.  Since the outer angle of suspension vertices
exceeds $\pi$, suspensions have no flat angle. Since each face $f$
(including the outer face) is a triangle, each face has precisely
three angles that are not flat. In other words every face $f$ has
$|f|-3$ incident vertices that are assigned to $f$.  This motivates
the definition:

\begin{definition}[FA Assignment] 
A \emph{flat angle assignment}
(FAA) is a mapping from a subset $U$ of the non-suspension vertices to
faces such that

\Item{{\rm (C$_v$)}} Every vertex of $U$ is assigned to at most one face,
\Item{{\rm (C$_f$)}} For every face $f$, precisely $|f|-3$ vertices are
assigned to $f$.

\end{definition}

Not every FAA induces an SLTR. An example is given in
Figure~\ref{f_bad-faa}.  Hence, we have to identify another
condition. To state this we need a definition. Let~$H$ be a connected
subgraph of the plane graph $G$. The \textit{outline cycle $\gamma(H)$
of $H$} is the closed walk corresponding to the outer face of $H$.  An
\emph{outline cycle} of $G$ is a closed walk that can be obtained as
outer cycle of some connected subgraph of~$G$. Outline cycles may have
repeated edges and vertices, see Figure~\ref{f_outlc}. The 
interior~$\int(\gamma)$ of an outline cycle $\gamma=\gamma(H)$ consists of $H$
together with all vertices, edges and faces of $G$ that are contained
in the area enclosed by $\gamma$.

\bigskip

\noindent\begin{minipage}[t]{0.5\textwidth} 
\centering
	\includegraphics[width=0.9\linewidth]{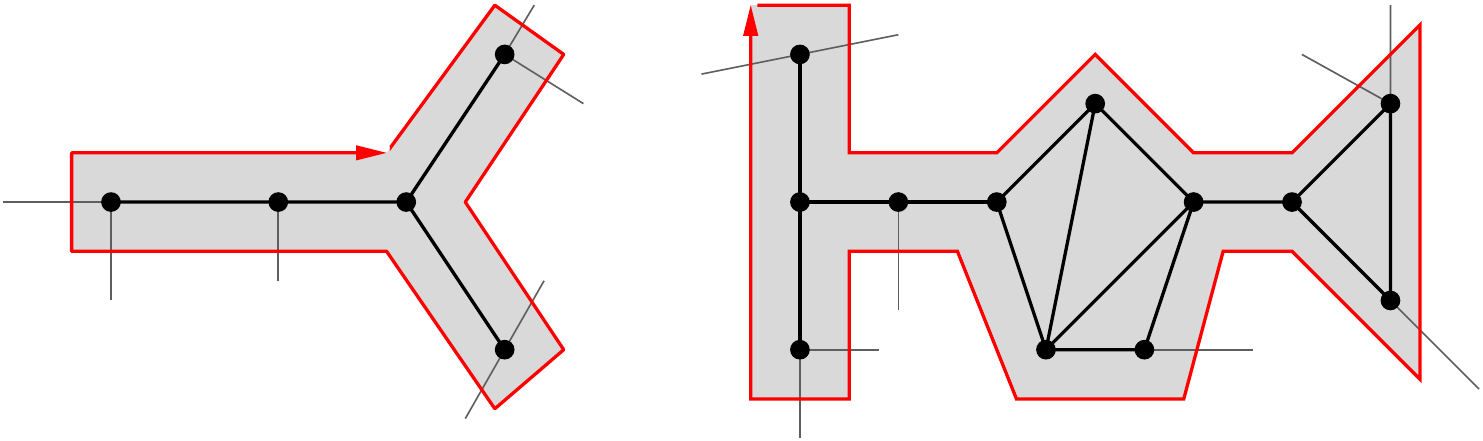}
	\captionof{figure}{\small Examples of outline cycles}\label{f_outlc}
\end{minipage}%
\begin{minipage}[t]{0.5\textwidth} 
\centering
	\includegraphics[width=\linewidth]{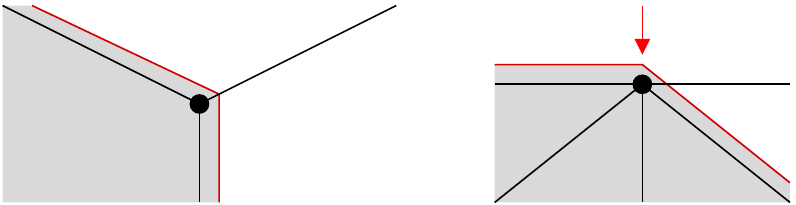}
	\captionof{figure}{\small Combinatorially Convex
Corners}\label{f_combconv}
\end{minipage}

\bigskip

\begin{proposition}\label{p_necFAA}
An SLTR obeys the following condition C$_o$:
\Item{{\rm (C$_o$)}} Every outline cycle that is not the outline cycle of a
path, has at least three geometrically convex corners.
\end{proposition}

\begin{proof}
  Consider an SLTR. Suppose that there is a connected subgraph, not a
  path, such that its outline cycle has less than three geometric
  convex corners. If the outline cycle has at most two geometric
  convex corners, then the subgraph is mapped to a line in the
  plane. The subgraph must either contain a vertex of degree more than
  three, or a face. If a vertex~$v$ together with three its neighbors
  is mapped onto a line, then the boundary of at least one of the
  faces incident to~$v$ is not a triangle. On the other hand if the
  subgraph contains a face, then this face is mapped to a line and,
  therefore, its boundary is not a triangle. In both cases the
  properties of an SLTR are violated.  This shows that~C$_o$ is a
  necessary condition.
\end{proof}

Condition C$_o$ has the disadvantage that it depends on a given SLTR,
hence, it is useless for deciding whether a planar graph $G$ admits an
SLTR.  The following Definition~allows to replace C$_o$ by a
combinatorial condition on an FAA.

\begin{definition}
Given an FAA $\psi$. A vertex $v$ of an outline cycle $\gamma$ is a
\emph{combinatorial convex corner} for~$\gamma$ with respect to $\psi$
if
\Item{{(K1)}} $v$ is a suspension vertex, or
\Item{{(K2)}} $v$ is not assigned and there is an edge $e$ incident
to $v$ with $e \not\in \int(\gamma)$, or
\Item{{(K3)}} $v$ is assigned to a face $f$, $f\not\in
\int(\gamma)$ and there exists an edge $e$ incident to $v$ with
$e\not\in \int(\gamma)$.
\end{definition} 

In Figure~\ref{f_combconv} an unassigned and an assigned combinatorially
convex corner are shown. The grey area represents the interior of some
outline cycle and the arrow represents the assignment of the vertex to
the face in which the arrow is drawn.

\begin{proposition}\label{p_geomcomb} 
Let $G$ admit an SLTR $\Gamma$,
that induces the FAA $\psi$ and let $H$ be a connected subgraph of $G$. 
If $v$ is a geometrically convex corner of the outline cycle
$\gamma(H)$ in $\Gamma$, then $v$ is a combinatorially convex corner
of $\gamma(H)$ with respect to $\psi$.
\end{proposition}
\begin{proof}
If $v$ is a suspension vertex it is clearly geometrically and
combinatorially convex.

Let $v$ be geometrically convex and suppose that $v$ is not a
suspension and not assigned by $\psi$.  In this case $v$ is interior
and, with respect to $\gamma$, the outer angle at $v$ exceeds $\pi$.
Therefore, at least two incident faces of $v$ are in the outside of
$\gamma$. These faces can be chosen to be adjacent, hence, the edge
between them is an edge $e$ with~$e\not\in \int(\gamma)$.  This shows
that $v$ is combinatorially convex.

Let $v$ be geometrically convex and suppose that $v$ is assigned to
$f$ by $\psi$. If $f\in \int(\gamma)$, then the inner angle of $v$
with respect to $\gamma$ is at least $\pi$. This contradicts the fact
that~$v$ is geometrically convex. Hence $f\not\in \int(\gamma)$. If
there is no edge $e$ incident to~$v$ such that~$e\not\in
\int(\gamma)$, then $v$ has an angle of size $\pi$ with respect to
$\gamma$. This again contradicts the fact that $v$ is geometrically
convex. Therefore, if $v$ is geometrically convex and assigned to $f$,
then $f\not\in \int(\gamma)$ and there exists an edge $e$ incident to
$v$ such that~$e\not\in \int(\gamma)$. This shows that $v$ is a
combinatorial convex corner for $\gamma$.
\end{proof}

The Proposition~enables us to replace the condition on geometrically
convex corners w.r.t.\ an SLTR by a condition on combinatorially
convex corners w.r.t.\ an FAA.

\Item{{\rm (C$^*_o$)}} Every outline cycle that is not the outline cycle of
a path, has at least three combinatorially convex corners.

From Proposition~\ref{p_necFAA} and Proposition~\ref{p_geomcomb} it follows that
this condition is necessary for an FAA that induces an SLTR.  
In Theorem~\ref{t_main} we prove that if an FAA obeys~C$^*_o$ then it
induces an SLTR. The proof is constructive.  In anticipation of
this result we say that an FAA obeying C$^*_o$ is a {\it good flat
angle assignment} and abbreviate it as a \emph{GFAA}.

Next we show that a GFAA induces a contact family of
pseudosegments. This family of pseudosegments is later shown to be
stretchable, i.e., it is shown to be homeomorphic to a contact system
of straight line segments.

\begin{definition} 
A \emph{contact family of pseudosegments} is a
family $\{c_i\}_i$ of simple curves
$c_i:[0,1]\rightarrow \mathbb{R}^2, \textrm{ with
different endpoints, i.e., }c_i(0)\neq c_i(1)$,
such that any two curves $c_j$ and $c_k$ ($j\neq k$)
have at most one point in common. If so, then this point is an
endpoint of (at least) one of them.
\end{definition}

A GFAA $\psi$ on a graph $G$ gives rise to a relation $\rho$ on the edges: Two
edges, incident to a common vertex $v$ and a common face $f$ are in relation
$\rho$ if and only if $v$ is assigned to $f$. The transitive closure of $\rho$
is an equivalence relation on the edges of $G$.
\begin{proposition}\label{p_GFAAtoPS} 
The equivalence classes of edges
of~$G$ defined by~$\rho$ form a contact family of pseudosegments.
\end{proposition}

\begin{proof} 
Let the equivalence classes of~$\rho$ be called arcs.

Condition C$_v$ ensures that every vertex is interior to at most one
arc. Hence, the arcs are simple curves and no two arcs cross.

An arc closing to a cycle yields an outline cycle that has no combinatorially
convex corner.  If an arc touches itself, then by C$_v$ it ends on itself. The
outline cycle of this equivalence class has at most one combinatorially
convex corner. Both cases contradict C$^*_o$.

If two arcs share two points, the outline cycle of the union has at most two
combinatorially convex corners. This again contradicts C$^*_o$.

We conclude that the family of arcs satisfies the properties of a
contact family of pseudosegments.  
\end{proof}

\begin{definition}\label{d_freept} 
Let $\Sigma$ be a family of
pseudosegments and let $S$ be a subset of $\Sigma$.  A~point~$p$ of a
pseudosegment from $S$ is a \emph{free point} for $S$ if
\Item{{(F1)}} $p$ is an endpoint of a pseudosegment in $S$, {\rm and}
\Item{{(F2)}} $p$ is not interior to a pseudosegment in $S$, {\rm and}
\Item{{(F3)}} $p$ is incident to the unbounded region of $S$, {\rm and}
\Item{{(F4)}} $p$ is incident to the unbounded region of $\Sigma$ {\rm or}\\
          $p$ is incident to a pseudosegment
that is not in $S$.
\end{definition}

\noindent With Lemma~\ref{l_CtoFP} we prove that the family of
pseudosegments $\Sigma$ that arises from a GFAA has the following
property%
\Item{{\rm (C$_P$)}} Every subset $S$ of $\Sigma$ with $|S|\geq 2$ has at
least three free points.

\begin{lemma}\label{l_CtoFP} 
Let $\psi$ a GFAA on a plane, internally
3-connected graph $G$. For every subset $S$ of the family of
pseudosegments associated with $\psi$, it holds that, if $|S|\geq 2$
then $S$ has at least 3 free points.
\end{lemma}

\begin{proof} 
Let $S$ be a subset of the contact family of
pseudosegments defined by the GFAA (Proposition~\ref{p_GFAAtoPS}).

Each pseudosegment of $S$ corresponds to a path in $G$. Let $H$ be the
subgraph of $G$ obtained as union of the paths of pseudosegments in
$S$.  We assume that $H$ is connected and leave the discussion of the cases
where it is not to the reader. If $H$ itself is not a path, then by
C$^*_o$ the outline cycle $\gamma(H)$ must have at least three
combinatorially convex corners. Every combinatorially convex corner of
$\gamma(H)$ is a free point of $S$.

If $S$ induces a path, then the two endpoints of this path are free
points for~$S$. Moreover, there exists at least one vertex $v$ in this
path which is an endpoint for two pseudosegments and not an interior
point for any. Now there must be an edge $e$ incident to $v$, such
that $e\not\in S$, therefore, $v$ is a free point for $S$.
\end{proof}

\begin{wrapfigure}[12]{r}{0.50\textwidth}%
\centering
\includegraphics[width=0.4\textwidth]{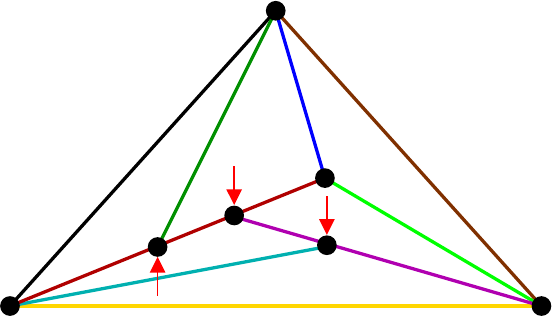}
\captionof{figure}{\small A stretched representation of a 
contact family of pseudosegments that arises from a GFAA.}
\label{fig:psegs}
\end{wrapfigure}

Given an internally 3-connected, plane graph $G$ with a GFAA. To find
a corresponding SLTR we aim at representing each of the pseudosegments
induced by the FAA as a straight line segment.  If this can be done,
every assigned vertex will be between its two neighbors that are part
of the same pseudosegment. This property can be modeled by requiring
that the coordinates~$p_v=(x_v,y_v)$ of the vertices of $G$ satisfy a
harmonic equation at each assigned vertex.

Indeed if $uv$ and $vw$ are edges belonging to a pseudosegment $s$,
then the coordinates satisfy 
\begin{equation} 
x_v = \lambda_v x_u + (1- \lambda_v)x_w
\qquad\mathrm{and}\qquad y_v = \lambda_v y_u + (1- \lambda_v)y_w
\label{eq:lin}
\end{equation} 
For some $\lambda_v$. In our model we can choose $\lambda_v$ as a parameter
from $(0,1)$. With fixed $\lambda_v$ the equations of (\ref{eq:lin}) are the
harmonic equations for $v$.
 
In the SLTR every unassigned vertex $v$ is placed in the convex hull
of its neighbors. In terms of coordinates this means that there are
$\lambda_{vu} > 0$ with $\sum_{u\in N(v)}
\lambda_{vu} = 1$ such that 
\begin{eqnarray} 
x_v = \sum_{u\in N(v)} \lambda_{vu} x_u, \qquad y_v =
\sum_{u\in N(v)}\lambda_{vu} y_u \,.
\label{eq:conv}
\end{eqnarray} 

Again for the model we can choose the $\lambda_{vu} > 0$ arbitrarily
subject to $\sum_{u\in N(v)}
\lambda_{vu} = 1$. With fixed parameters the equations (\ref{eq:conv}) 
enforce that $v$ is located in the a weighted
barycenter of its neighbors.
These are the harmonic equations for an unassigned vertex $v$. 

Vertices whose coordinates
are not restricted by harmonic equations are called \emph{poles}. In
our case the suspension vertices are the three poles of the harmonic
functions for the $x$ and $y$-coordinates. The coordinates for the
suspension vertices are fixed as the corners of some non-degenerate triangle,
this adds six equations to the linear system.

The theory of harmonic functions and applications to (plane) graphs are nicely
explained by Lov\'asz~\cite{lovasz09}. The proof of the following proposition is 
inspired by the proofs in Chapter~3 of~\cite{lovasz09}.
\begin{proposition}\label{prop:uniqueSol}
Let $G=(V,E)$ be a directed graph, $\lambda: E \to \RR^+$ 
be a weight function, and $P\subset V$ be a set of poles.
If every subset $Q$ of $V\setminus P$ has an out-neighbor 
in $V\setminus Q$, then for all $\psi_0:P\rightarrow \mathbb{R}$
there is an extension $\psi:V\rightarrow \mathbb{R}$ which is harmonic
on all $v \in V\setminus P$, i.e., 
$\psi(v) = \psi_0(v)$ for all $v\in P$ and 
$\psi(v) = \sum_{u \in out(v)} \lambda_{(v,u)}\psi(u)$ for all $v\in
V\setminus P$.
\end{proposition}
\begin{proof}

The proof has three steps, first we show that the maximum and minimum of a
harmonic function are attained at poles. Then we show that for
every map~$\psi_0:P\rightarrow \mathbb{R}$ from the set of poles to
the reals, there is a unique extension~$\psi:V\rightarrow
\mathbb{R}$ that is harmonic in all the vertices that are not
  poles. Last we show that a solution exists.

Let~$f$ be a non-constant harmonic function on $G$.
Let~$Q=\{v\in V: f(v) \textrm{ maximum}\}$ and~$Q'=\{v\in Q: v\textrm{ has an out-neighbor not in
}Q\}$. Since $f$ is not constant $Q\neq V$. Suppose $Q$ does not contain a pole. Since every subset $Q$ of $V\setminus P$ has an outneighbor in $V\setminus Q$, it follows that~$Q'$ is not empty. Elements of~$Q'$ are not
harmonic and, hence, must be poles poles, contradiction. Therefore, $Q$ must contain a pole. Similarly we
find a pole among the vertices where the minimum is attained.

Consider~$\psi_0:P\rightarrow \mathbb{R}$, a map from the set of
poles to the reals and suppose there are two
extensions~$\psi,\psi^*~:~V~\rightarrow~\mathbb{R}$ that satisfy the
harmonic equations of all non-poles. Then the
function~$\omega=\psi-\psi^*$ is also harmonic in all vertices not
in~$P$. As~$\psi$ and~$\psi^*$
are extensions of~$\psi_0$ the value of $\omega$ at all poles is zero.
Since maximum and minimum of a harmonic function are attained at poles,
we conclude that $\omega$ is zero everywhere, hence $\psi=\psi^*$.

Prescribed values at poles together with the harmonic equations at
non-poles, yield a linear system of $n$ equations in $n$ variables. 
From the uniqueness it follows that the homogeneous system has a
trivial kernel, hence, the system has a unique solution for 
every~$\psi_0:P\rightarrow \mathbb{R}$ prescribing values for the poles.
\end{proof}

To make use of Proposition~\ref{prop:uniqueSol} we need to show that a system of equations that comes from a GFAA, induces a directed graph and weight function that satisfy the above properties. The vertices of the directed graph are the vertices of $G$. For a vertex $v$ that is assigned and between $u$ and $w$, the edges $v \to u$ and $v\to w$ are added. For a not assigned vertex, a directed edge to every of its neighbors is added. The weights are given by the chosen parameters $\lambda_v$ and $\lambda_{vu}$. The poles are the suspension vertices. To show that every subset $Q$ of $V\setminus P$ has an out-neighbor in $V\setminus Q$, we consider the contact family of pseudosegments induced by the GFAA. 

Suppose there exists a non-empty set $Q\subseteq V\setminus P$ that has no out-neighbor in ${V\setminus Q}$.
Let $v$ be a vertex in $Q$. If $v$ is interior to a pseudosegment $p$, then all vertices of $p$ are in $Q$. If $v$ is not assigned, then all of its neighbors must be in $Q$. Therefore, $Q$ contains at least two pseudosegments. Moreover, $Q$ is not the whole set, as $Q\subseteq V\setminus P$. Since the contact family of pseudosegments comes from a GFAA, the set of pseudosegments contained in $Q$ must have at least three free points. A free point is on the boundary, not interior to any pseudosegment in $Q$ and has at least one neighbor outside $Q$. Therefore, $Q$ must have an out-neighbor in $V\setminus Q$. 

Now we state our main result, it shows that the necessary conditions
are also sufficient.
\begin{theorem}\label{t_main} 
Let $G$ be an internally 3-connected, plane
graph and $\Sigma$ a family of pseudosegments associated to an
FAA, such that each subset $S\subseteq \Sigma$ has three free points
or cardinality at most one. The unique solution of the system of
equations that arises from $\Sigma$ is an SLTR.
\end{theorem}
\noindent{\it Proof.}
The proof consists of 7 arguments, which together yield that the drawing
induced from the GFAA is a non-degenerate, plane drawing.  The proof has been
inspired by proof for convex straight line drawings of plane graphs via spring
embeddings shown to us independently by G\"unter Rote and \'{E}ric Fusy, both
attribute key ideas to \'{E}ric Colin de Verdi\`ere.

\smallskip

\noindent 1.\ \textit{Pseudosegments become Segments.}  
Let $(v_1,v_2),(v_2,v_3),\ldots,(v_{k-1},v_k)$ be the set of edges of
a pseudosegment defined by $\psi$. The harmonic conditions for the
coordinates force that $v_i$ is placed between $v_{i-1}$ and $v_{i+1}$
for $i=2,..,k-1$.  Hence all the vertices of the pseudosegment are
placed on the segment with endpoints $v_1$ and $v_k$.

\smallskip

\noindent 2.\ \textit{Convex Outer Face.} 
The outer face is bounded by three pseudosegments and the suspensions
are the endpoints for these three pseudosegments. The coordinates of
the suspensions (the poles of the harmonic functions) have been chosen
as corners of a non-degenerate triangle and the pseudosegments are
straight line segments, therefore, the outer face is a triangle and in
particular convex.

\smallskip

\noindent 3.\ \textit{No Concave Angles.} 
Every vertex, not a pole, is forced either to be on the line segment
between two of its neighbors (if assigned) or in a weighted barycenter
of all its neighbors (otherwise). Therefore, every non-pole vertex is
in the convex hull of its neighbors. This implies that there are no
concave angles at non-poles.

\smallskip

\noindent 4.\ \textit{No Degenerate Vertex.} 
A vertex is degenerate if it is placed on a line, together with at
least three of its neighbors. Suppose there exists a vertex $v$, such
that $v$ and at least three of its neighbors are placed on a line
$\ell$. Let $S$ be the connected component of pseudosegments that are
aligned with $\ell$, such that $S$ contains $v$. The set $S$ contains at
least two pseudosegments. Therefore, $S$ must have at least three free
points, $v_1,v_2,v_3$.

By property 4 in the Definition~of free points, each of the free
points is incident to a segment that is not aligned with $\ell$.  Suppose
the free points are not suspension vertices. If $v_i$ is interior to
$s_i\in S$, then $s_i$ has an endpoint on each side of $\ell$. If $v_i$ is not
assigned by the GFAA it is in the strict convex hull of its neighbors,
hence,~$v_i$ is an endpoint of a segment reaching into each of the two
half-planes defined by~$\ell$.

Now suppose $v_1$ and $v_2$ are suspension vertices. Since not all three
suspension vertices lie on one line, at least one of the three free points is
not a suspension. Let $v_3$ be such a free point.  If $v_3$ is interior to a
pseudosegment not on $\ell$, then one endpoint of this pseudosegment lies
outside the convex hull of the three suspensions, which is a
contradiction. Hence it is not interior to any pseudosegment and at least one
of its neighbors does not lie on $\ell$, but then $v_3$ should be in a
weighted barycenter of its neighbors, hence again we would find a vertex
outside the convex hull of the suspension vertices. Therefore, at most one of
the free points is a suspension and $\ell$ is incident to at most one of the
suspension vertices.

In any of the above cases each of $v_1,v_2,v_3$ has a neighbor on either side of~$\ell$.

Let $n^+$ and $n^-=-n^+$ be two normals for line $\ell$ and let $p^+$ and
$p^-$ be the two poles, that maximize the inner product with $n^+$
resp.~$n^-$.  Starting from the neighbors of the $v_i$ in the positive
halfplane of $\ell$ we intend to move to a neighbor with
larger  inner product with $n^+$ until we
reach $p^+$.  If $n^+$ is perpendicular to another segment this may
not be possible. In this case, however, we can use a slightly perturbed vector
$n_\varepsilon^+$ to break ties and make the intended progress towards $p^+$
possible.

Hence $v_1,v_2,v_3$ have paths to $p^+$ in the upper
halfplane of $\ell$ and paths to $p^-$ in the lower halfplane.  
Since~$v_1,v_2,v_3$ also have a path to $v$ we can contract all vertices of
the upper and lower halfplane of $\ell$ to $p^+$ resp.~$p^-$ and all
inner vertices of these paths to $v$ to produce a $K_{3,3}$ minor of
$G$. This is in contradiction to the planarity of $G$. Therefore,
there is no degenerate vertex.

\smallskip

\noindent 5.\ \textit{Preservation of Rotation System.}  Let $\theta(v) =
\sum_{f} \theta(v,f)$ denote the sum of the angles around an inner vertex.
Here $f$ is a face incident to $v$ and $\theta(v,f)$ is the (smaller!) angle
between the two edges incident to $v$ and $f$ in the drawing obtained by
solving the harmonic system. If the incident faces are oriented consistently
around $v$, then the angles sum up to $2\pi$. In general there may be some
folding, see Figure~\ref{f_presrot} but we can argue that this increases the
angle sum.  Indeed $v$ has three neighbors $x,y,z$ such that every closed
halfspace containing $v$ also contains one of these three. The angular sum to
get from $x$ via $y$ to $z$ is at least the larger of the two angles 
between $x$ and $z$, i.e., some $\rho \geq \pi$. The angular sum to get back
from $z$ to $x$ is at least $2\pi - \rho$ or if it again included a visit
at $y$ at least $\rho$. In either case the angular sum exceeds $2\pi$,
i.e., $\theta(v)\geq\pi$ for all inner vertices $v$.
%
\begin{figure}[ht] \centering
	\includegraphics[width=0.6\linewidth]{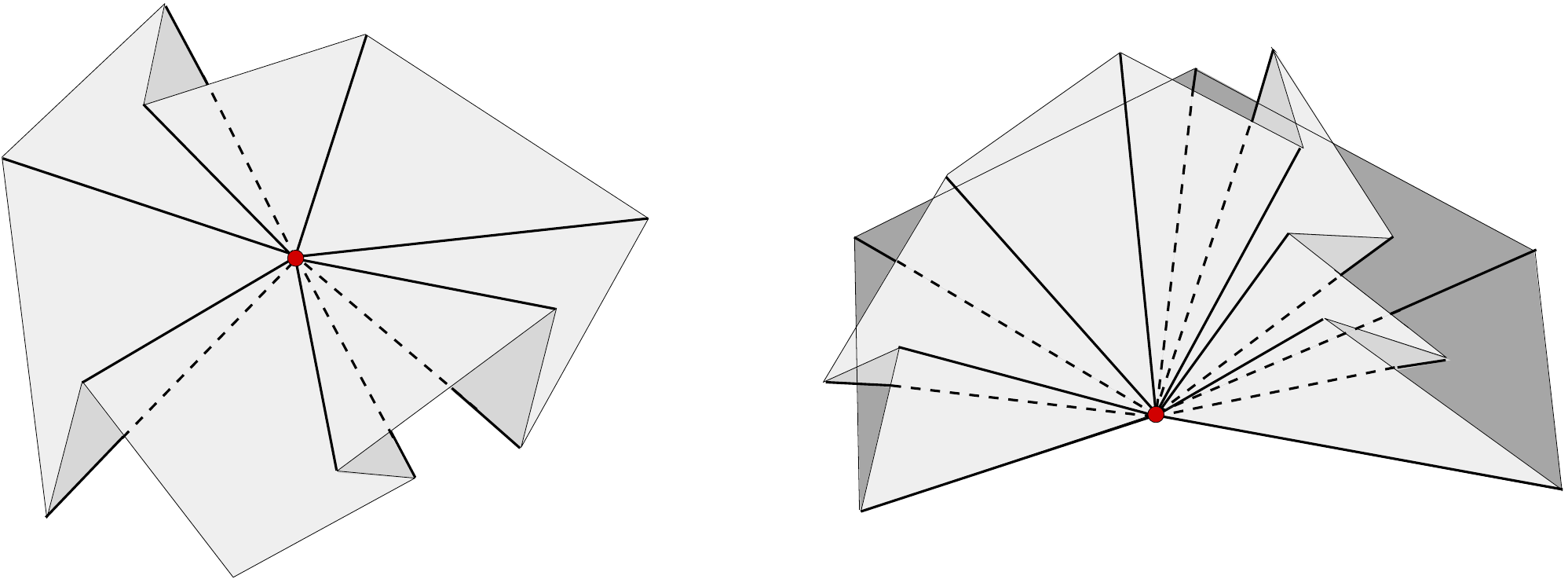}
	\caption{\small If the incident faces are not oriented consistently
around $v$, then the angles sum up to more than $2\pi$.}\label{f_presrot}
\end{figure} 

We do not include the outer face in the sums so that the $b$
vertices incident to the outer face contribute a total angle 
of at least $(b-2)\pi$ to the inner faces.

 
Now consider the sum $\theta(f)= \sum_{v} \theta(v,f)$ of the angles
of a face $f$. A triangulation of the face $f$ in the planar
drawing consists of $|f|-2$ triangles. The angle sum of these triangles 
in the straight line is $(|f|-2)\pi$. The angles of the triangles 
incident to $v$ cover at least the smaller of the two angles formed
by the two edges incident to $v$ and $f$. 
Hence, $(|f|-2)\pi \geq \theta(f)$.

The sum over all
vertices $\sum_{v} \theta(v)$ and the sum over all faces $\sum_{f}
\theta(f)$ must be equal since they count the same angles in two
different ways.
\begin{equation} (|V|-b)2\pi + (b-2)\pi \leq \sum_{v}
\theta(v)=\sum_{f} \theta(f)\leq \left( (2|E|-b)-2(|F|-1)\right)\pi
\end{equation} 
This yields $|V|-|E|+|F|\leq 2$. Since $G$ is planar
Euler's formula implies equality. Therefore, $\theta(v)=2\pi$ for every
interior vertex $v$ and the faces must be oriented consistently around
every vertex, i.e. the rotation system is preserved. Note that the
rotation system could have been flipped, between clockwise and
counterclockwise but then it is flipped at every vertex.

\smallskip

\noindent 6.\ \textit{No Crossings.} 
Suppose two edges cross. On either side of both of the edges there is
a face, therefore, there must be a point $p$ in the plane which is
covered by at least two faces. Outside of the drawing there is only
the unbounded face. Move along a ray, that does not pass through a
vertex of the graph, from $p$ to infinity. A change of the cover
number, i.e. the number of faces by which the point is covered, can
only occur when crossing an edge. But if the cover number changes
then the rotation system at a vertex of that edge must be wrong.  
This would contradict the previous item.  Therefore, a crossing cannot exist.

\smallskip

\noindent 7.\ \textit{No Degeneracy.} 
Suppose there is an edge of length zero. Since every vertex has a path
to each of the three suspensions there has to be a vertex $a$ that is
incident to an edge of length zero and an edge $ab$ of non-zero
length. Following the direction of forces we can even find such a
vertex-edge pair with $b$ contributing to the harmonic equation for
the coordinates of $a$.  We now distinguish two cases.

\noindent If $a$ is assigned, it is on the segment between $b$ and some
$b'$, together with the neighbor of the zero length edge this makes
three neighbors of $a$ on a line. Hence, $a$ is a degenerate vertex. A
contradiction.

\noindent If $a$ is unassigned it is in the convex hull of its
neighbors.  However, starting from $a$ and using only zero-length
edges we eventually reach some vertex $a'$ that is incident to an edge
$a'b'$ of non-zero length, such that $b'$ is contributing to the
harmonic equation for the coordinates of $a'$. Vertex $a'$ has the
same position as $a$ and is also in the convex hull of its neighbors.
This makes a crossing of edges unavoidable. A contradiction.  Hence,
there are no edges of length zero.

\noindent Suppose there is an angle of size zero. Since every vertex
is in the convex hull of its neighbors there are no angles of size
larger than $\pi$.  Moreover there are no crossings, hence the face
with the angle of size zero is stretching along a line segment with
two angles of size zero.  Since there are no edges of length zero and
all vertices are in the convex hull of their neighbors, all but two
vertices of the face must be assigned to this face. Therefore, there
are two pseudosegments bounding this face, which have at least two
points in common, this contradicts that $\Sigma$ is a family of
pseudosegments.  We conclude that there is no degeneracy.
\smallskip

\noindent
From 1--7 we conclude that the drawing is plane and
thus an SLTR.  
\qed

%
For later use we will show that it is sufficient 
to verify condition C$^*_o$ for outline cycles that are simple outline cycles,
i.e., outline cycles without cut vertices.
 
\begin{lemma}\label{lem:simpleCycles}
  Given a planar 3-connected graph $G$ and an FAA such that every simple
  outline cycle has at least three combinatorially convex corners. Then every
  outline cycle, not the outline cycle of a path, has at least three
  combinatorially convex corners.
\end{lemma}
\begin{proof}
Suppose the Lemma~does not hold. Let $\bar{\gamma}$ the smallest outline cycle, not the outline cycle of a path, that has at most two combinatorially convex corners. Let $\gamma$ the largest simple outline cycle contained in~$\bar{\gamma}$.

Suppose $\gamma$ contains only one vertex. As $\bar{\gamma}$ is not the outline cycle of a path, there exists a $v\in\bar{\gamma}$ which has degree at least three in $\bar{\gamma}$, let $\gamma = \{v\}$. Now $\bar{\gamma}-\gamma$ has at least three components, let $C$ be such a component. If $|C|=1$ then this vertex is a combinatorially convex corner for $\bar{\gamma}$. If $C$ is a path then (at least) the vertex that is not connected to $v$ is a combinatorially convex corner for $\bar{\gamma}$. If $C$ is not a path, then since it is smaller than $\bar{\gamma}$, it has at least three combinatorially convex corners. At least two of those must also be combinatorially convex corners of $\bar{\gamma}$. We conclude that when $\gamma$ contains only one vertex, $\bar{\gamma}$ has at least three combinatorially convex corners.

Suppose $\gamma$ is a cycle of length at least three. As $\bar{\gamma}$ is not a simple outline cycle, $\bar{\gamma}-\gamma$ has at least one component. Such a component can connect to at most one vertex of $\gamma$ as otherwise $\gamma$ is not the largest simple outline cycle in $\bar{\gamma}$. Similar as in the previous case, each component in $\bar{\gamma}-\gamma$ contributes at least one combinatorially convex corner. As $\gamma$ has at least three combinatorially convex corners, it now follows that $\bar{\gamma}$ has at least three combinatorially convex corners. This concludes the proof.
\end{proof}

\section{Further Applications of the Proof Technique}\label{applications}

We have shown that a graph $G$ has an SLTR exactly
if it admits an FAA satisfying C$_v$, C$_f$ and C$^*_o$.
Conditions C$_v$ and C$^*_o$ are necessary for the proof that
the system of pseudosegments corresponding to the FAA is
stretchable. Condition C$_f$, however, is only needed to make all 
the faces triangles. Modifying condition C$_f$ allows for
further applications of the stretching technique. 

We still need that least three poles (suspensions)  in
convex position. 
Also we have to make sure that no vertex of the outer face is
assigned to an inner face. And of  course we still need at least
three corners for every face.
Together this makes the modified face condition:
\Item{{\rm (C$_f^*$)}} For every face $f$, at most $|f|-3$ vertices are
assigned to $f$ and no
vertex of the outer face $f^o$ are assigned to an inner face.
\smallskip

\noindent
If we use the empty flat angle assignment, i.e., if the 
harmonic equations of all non-suspensions are of type~(\ref{eq:conv}),
then  we obtain a drawing such that all non-suspension vertices are
in the barycenter of their neighbors. If all vertices from the outer face are
suspensions, this is the 
Tutte drawing with asymmetric elastic forces given by the 
parameters $\lambda_{uv}$, see ~\cite{tutte62} and~\cite{lovasz09}. 
Note that in this case the existence of at least three combinatorially
convex corners at an outline cycle (condition C$^*_o$) follows
from the internally 3-connectedness of the graph.
\smallskip

\noindent
The  construction of Section~\ref{SLTR} also
applies when 
\Bitem the assignment has $|f|-i$ vertices
assigned to every inner face $f$, for $i=4,5$ (drawing with only convex 
4-gon or only convex 5-gon faces.) 
\Bitem the assignment has some number $c_f$ of 
corners at inner face $f$ (drawing with convex 
faces of prescribed complexity).
\smallskip

\noindent
The drawback is that again in these cases we do not know how to 
find an FAA that fulfills C$^*_o$.

In~\cite{kenshef04} Kenyon and Sheffield
study $T$-graphs in the context of dimer configurations (weighted
perfect matchings). In our terminology $T$-graphs correspond to
straight line representations such that each
non-suspension is assigned. In~\cite{kenshef04} the straight line
representations of  $T$-graphs are obtained by analyzing
random walks. Cf.~\cite{lovasz09} for further connections between
discrete harmonic functions and Markov chains.
 
\medskip

\noindent\textbf{Stretchability of Systems of Pseudosegments.}  
A contact system of pseudosegments is {\it stretchable} if it is
homeomorphic to a contact system of straight line segments.  De
Fraysseix and Ossona de Mendez characterized stretchable systems of
pseudosegments~\cite{dfdm05}. They use the notion of
an extremal point.

\begin{definition}\label{d_extrpt} 
Let $\Sigma$ be a family of
pseudosegments and let $S$ be a subset of $\Sigma$.  A point $p$ is an
\emph{extremal point} for $S$ if
\Item{{(E1)}}  $p$ is an endpoint of a pseudosegment in $S$, {\rm and}
\Item{{(E2)}}  $p$ is not interior to a pseudosegment in $S$, {\rm and}
\Item{{(E3)}}  $p$ is incident to the unbounded region of $S$.
\end{definition}

\begin{theorem}[De Fraysseix \& Ossona de Mendez~\hbox{\rm \cite[Theorem~38]{dfdm05}}]
\label{t_dFdM1}
A contact family $\Sigma$ of pseudosegments is stretchable if and only
if each subset $S \subseteq \Sigma$ of pseudosegments with $|S|\geq
2$, has at least 3 extremal points.
\end{theorem}

Our notion of a free point (Definition~\ref{d_freept}) contains the three properties
of an extremal point but adds a fourth condition. In the following we show
that there is no big difference. First in Proposition~\ref{p_dfdm} we show that in
the case of families of pseudosegments that live on a plane graph via an FAA,
the two notions coincide. Then we continue by reproving Theorem~\ref{t_dFdM1} as
a corollary of Theorem~\ref{t_main}. The proof of Theorem~\ref{t_dFdM1}
in~\cite{dfdm05} is based on a long and complicated inductive construction.

\begin{proposition}\label{p_dfdm} 
  Let $G$ be an internally 3-connected, plane graph and $\Sigma$ a
  family of pseudosegments associated to an FAA, such that each subset
  $S\subseteq \Sigma$ has three extremal points or cardinality at most
  one. The unique solution of the system of equations corresponding
  to $\Sigma$, is an SLTR.
\end{proposition}
\noindent{\it Proof.}
Note that in the proof of Theorem~\ref{t_main} the notion of free points is only
used to show that there is no degenerate vertex.  We show how to modify this
part of the argument for the case of extremal points:

Consider again the set $S$ of pseudosegments aligned with $\ell$. We will
show that all extremal points are also free
points. Let $p$ be an extremal point of $S$. Assuming that $p$ is not
free, we can negate condition~4.\ from Definition~\ref{d_freept}, i.e., all the
pseudosegments for which $p$ is an endpoint are in~$S$. Since $p$ is not
interior to a pseudosegment in $S$ it follows from 
3-connectivity that $p$ is incident to at least three pseudosegments, all of
which lie on the line $\ell$. Since all regions are bounded by three
pseudosegments and $p$ is not interior to a segment of $S$, all the
regions incident to $p$ must lie on $\ell$. But then $p$ is not incident
to the unbounded region of $S$, hence $p$ is not an extremal point.
Therefore, all extremal points of~$S$ are also free points of~$S$.
Proposition~\ref{p_dfdm} now follows from Theorem~\ref{t_main}.\qed

\vbox{}
\noindent{\it Proof (of Theorem~\ref{t_dFdM1}).}\/ Let $\Sigma$ a contact family
of pseudosegments which is stretchable. Consider a set $S\subseteq \Sigma$ of
cardinality at least two in the stretching, i.e., in the segment
representation.  Endpoints (of segments) on the boundary of the convex hull of
$S$ are extremal points.  There are at least three of them unless $S$ lies on
a line $\ell$. In the collinear case, there is a point~$q$ on $\ell$ that is
the endpoint of two segments for $S$. This is a third extremal point.

Conversely, assume that each subset $S \subseteq \Sigma$ of
pseudosegments, with $|S|\geq 2$, has at least 3 extremal points. 
We aim at applying Prop~\ref{p_dfdm}. To this end we
construct an extended system $\Sigma^+$ of pseudosegments in
which every region is bounded by precisely three pseudosegments.

First we take a set $\Delta$ of three pseudosegments that intersect like the
three sides of a triangle so that $\Sigma$ is in the interior. The corners of
$\Delta$ are chosen as suspensions and the sides of $\Delta$ are deformed such
that they contain all extremal points of the family $\Sigma$. Let the new
family be $\Sigma'$.

\begin{figure}[ht]%
\centering
\includegraphics[width=0.4\textwidth]{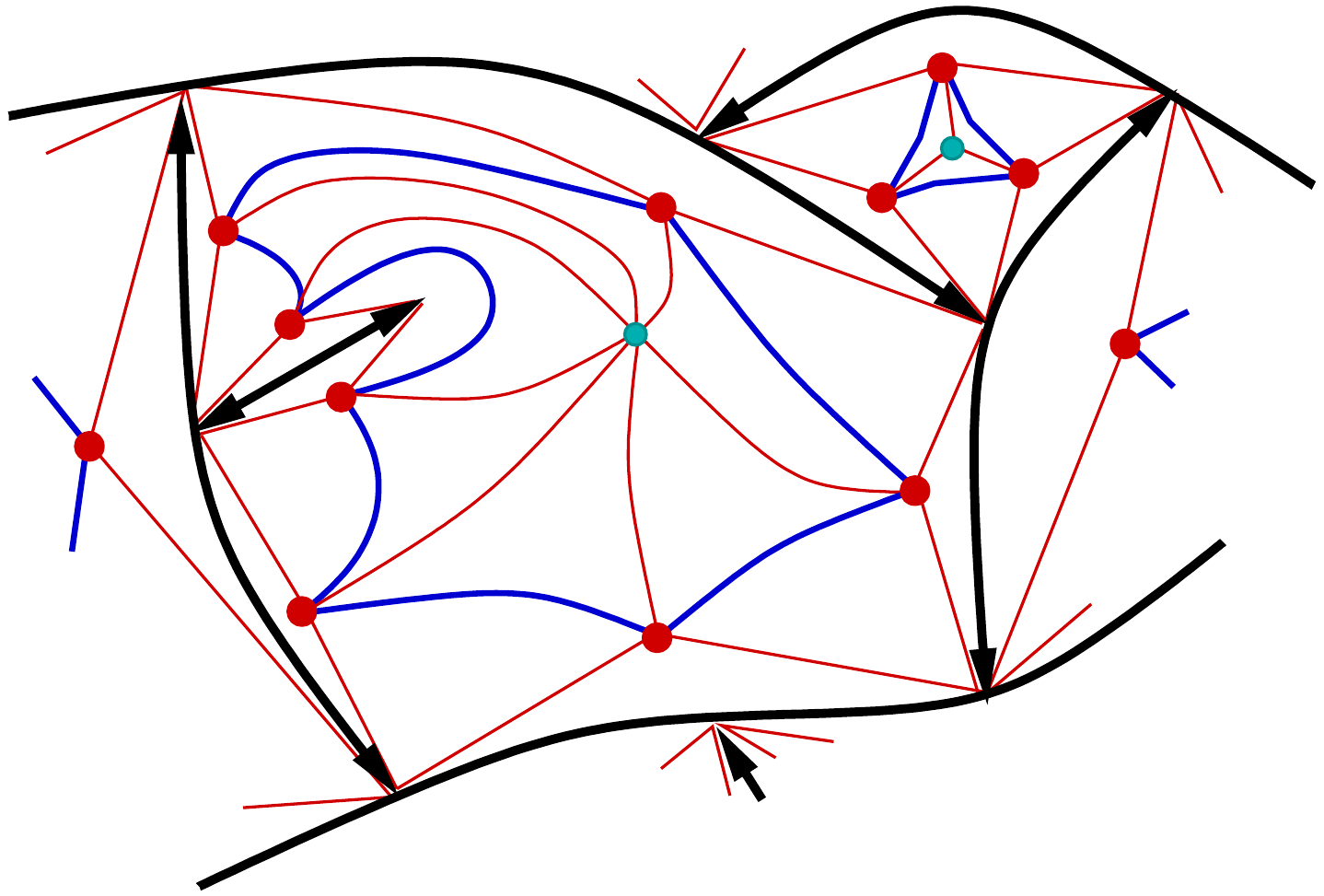}
\caption{\small Protection points in red and the triangulation
point in cyan for two faces of some~$\Sigma'$.} 
\label{fig:protP}
\end{figure}
Next we add \emph{protection points}, these additional points ensure
that the pseudosegments of $\Sigma'$ will be mapped to straight lines.
For each inner region $R$ in $\Sigma'$, for each pseudosegment $s$ in $R$, we
add a protection point for each visible side of $s$. The protection
point is connected to the endpoints of $s$, with respect to $R$ from
the visible side of $s$.

Now the inner part of $R$ is bounded by an alternating sequence of
endpoints of $\Sigma'$ and protection points. We connect two protection
points if they share a neighbor in this sequence.  Last we add a
\emph{triangulation point} in $R$ and connect it to all protection points
of $R$.

This construction yields a family  $\Sigma^+$ of pseudosegments such that
every region is bounded by precisely three pseudosegments and every
subset $S\subseteq \Sigma^+$ has at least 3 extremal points, unless it
has cardinality one.%

Let $V$ be the set of points of $\Sigma^+$ and $E$ the set of 
edges induced by $\Sigma^+$.  It follows from the
construction that $G=(V,E)$ is internally 3-connected.

By Proposition~\ref{p_dfdm} the graph $G=(V,E)$ together with $\Sigma^+$ is
stretchable to an SLTR. Removing the protection points, triangulation
points and their incident edges yields a contact system of straight
line segments homeomorphic to $\Sigma$.  
\qed

\subsection{Schnyder Woods and Primal-Dual Contact Representations}

Schnyder woods were introduced in the context of order
dimension~\cite{Schnyder89}.  In a second publication Schnyder used
them for compact straight line drawings of planar
graphs~\cite{Schnyder90}. Schnyder woods have since found many
additional applications to various graph drawing models as well as to
the enumeration and encoding of planar maps.  The notion of Schnyder
woods was generalized to 3-connected planar graphs~\cite{Felsner01}.
Gon\c{c}alves, L{\'e}v{\^e}que and Pinlou~\cite{GoncalvesLP12} used
Schnyder woods of 3-connected planar graphs for the construction of
primal-dual contact representations with triangles. They proof that
each Schnyder wood induces a stretchable contact family of
pseudosegments which represents the primal-dual contact graph. In this
section we give a simpler proof of this result using geodesic
embeddings on orthogonal surfaces. The theory was again developed in
the context of order dimension~\cite{m-pgmrtmi-02,f-gepg-03,fz-swaos-08}.

\begin{definition}[Schnyder Wood]
  Let $G$ be a 3-connected plane graph with three suspensions $s_1,s_2,s_3$ in
  clockwise order on the boundary of the outer face. A Schnyder wood is an
  orientation and labeling of the edges of $G$ with the labels 1, 2 and 3 such
  that the following four conditions are satisfied\footnote{The labels are
    considered in a cyclic structure, such that $(i-1)$ and $(i+1)$ are always
    well defined.}.
\Item{{(S1)}} Each edge is either unidirected or bidirected. In the
  latter case the two directions have distinct labels.
\Item{{(S2)}} At each suspension $s_i$ there is an additional half edge with
  label $i$ pointing into the outer face.
\Item{{(S3)}} Each vertex $v$ has outdegree one in each label. Around
  $v$ in clockwise order there is an outgoing edge of label $1$, zero
  or more incoming edges of label 3, an outgoing edge of label 2, zero
  or more incoming edges of label 1, an outgoing edge of label 3 and
  zero or more incoming edges of label 2.
\Item{{(S4)}} There is no directed cycle in one color.
\end{definition}

\paragraph{Primal-Dual Triangle Contact representation.}
In a triangle contact representation of a graph, the vertices are
represented by a collection of interiorly disjoint triangles and edges
correspond to point-to-side contacts between the triangles. De Fraysseix,
Ossona de Mendez and Rosenstiehl proved that every planar graph has a triangle
contact representation~\cite{fray94}.

A primal-dual contact representation of a plane graph by triangles, is a
dissection of a triangle into triangles with a correspondence between  
the triangles of the dissection and the union of vertices and dual vertices
(faces) of the graph. Point contacts between triangles
correspond to edges of the graph and its dual, while side contacts
correspond to incidences between vertices and faces. The enclosing triangle
of the primal-dual contact representation corresponds to the outer face.
Note that a triangle contact representation of a triangulation immediately
yields a primal-dual contact representation, the only detail that needs
to be adjusted is that the outer face has to get triangular shape. 

Gon\c{c}alves, L{\'e}v{\^e}que and Pinlou have shown that every 3-connected
planar graph $G$ has a primal-dual contact representation by
triangles~\cite{GoncalvesLP12}. They use a Schnyder wood of the primal graph
to define a family of pseudosegments and then use the results of~\cite{dfdm05}
to show that this system is stretchable. Moreover they have shown that
primal-dual contact representations are in one-to-one correspondence with
Schnyder woods of planar 3-connected graphs.

\begin{figure}[ht]
\begin{center}
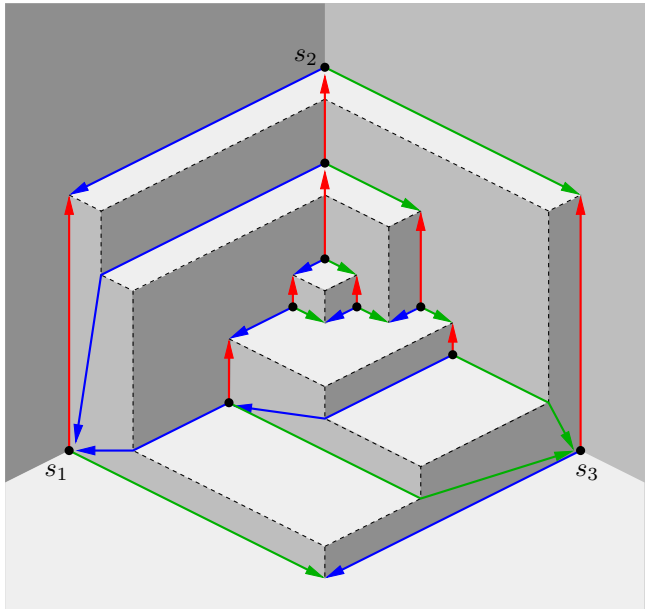
\end{center}
\caption{\small A geodesic embedding. The vertices of the graph are the local minima
  of the orthogonal surface. The edges carry the coloring and
  orientation of a Schnyder wood.}\label{fig:geoexample}
\end{figure}

We give a simpler proof of the first part. The proof is based on outline
cycles and a geodesic embedding of the graph.
To begin we need some definitions.

With a point $p\in \RR^d$ associate its {\em cone} $C(p) = \{ q\in \RR^d : p
\leq q \}$.  The {\em filter} $\langle \mathcal{V} \rangle$ generated by
a finite set $\mathcal{V} \subset \RR^d$ is the union of all cones $C(v)$ for $v \in
\mathcal{V}$.  The {\em orthogonal surface} $S_\mathcal{V}$ generated by
$\mathcal{V}$ is the boundary of $\langle \mathcal{V} \rangle$. A point $p \in
\RR^d$ belongs to $S_\mathcal{V}$ if and only if $p$ shares a coordinate with
all $v\leq p$, $v\in \mathcal{V}$.  The generating set $\mathcal{V}$ is an
antichain if and only if all elements of $\mathcal{V}$ appear as minima on
$S_\mathcal{V}$.  Figure~\ref{fig:geoexample} shows an example of an
orthogonal surface with an embedded graph.  The vertices of the graph are the
elements of $\mathcal{V}$. Each vertex is incident to three ridges, we call
them {\em orthogonal arcs}. The set of all orthogonal arcs of the surface yields the
partition into plane patches, we call them \emph{flats}.  An \emph{elbow
  geodesic} is a connection between two vertices $u$ and $v$, it connects
the two vertices with line segments on the surface to a saddle-point $s$ of
$S_\mathcal{V}$. One or both of the line segments forming an elbow
geodesic, are orthogonal arcs. 

Figure~\ref{fig:geoexample} shows a geodesic embedding, in fact the geodesic
embedding is decorated with the orientation and coloring of a Schnyder
wood. Miller~\cite{m-pgmrtmi-02} was the first to observe the connection
between Schnyder woods and orthogonal surfaces in $\RR^3$.

\begin{definition}[Geodesic Embedding] 
  Let $G$ a plane 3-connected graph. A drawing of $G$ onto an orthogonal
  surface $\mathcal{S}_\mathcal{V}$ generated by an antichain $\mathcal{V}$ is
  a geodesic embedding if the following axioms are satisfied.
\Item{{(G1)}} There is a bijection between the vertices of $G$ and the points in $\mathcal{V}$.
\Item{{(G2)}} Every edge of $G$ is an elbow geodesic in $\mathcal{S}_\mathcal{V}$ and every bounded orthogonal arc in $\mathcal{S}_\mathcal{V}$ belongs to an edge in $G$.
\Item{{(G3)}} There are no crossing edges in the embedding of $G$ on $\mathcal{S}_\mathcal{V}$.
\end{definition}

Let $G$ be a 3-connected plane graph with suspensions $a_1,a_2,a_3$ and
let $T$ be a Schnyder wood of $G$. There is an orthogonal surface $S$, such that,
$G$ has a geodesic embedding on $S$ that induces $T$. Taking the maxima of $S$
as vertices, we obtain a geodesic embedding of the dual $G^*$ of $G$ without the
vertex $v^*_\infty$ representing the outer face (edges
of $G^*$ connecting to $v^*_\infty$ are unbounded rays). The geodesic
embedding of $G^*$ is naturally decorated with colors and orientations.
Adding one suspension for the unbounded rays of each color, yields a Schnyder
wood $T^*$ of the dual. The pair $(T,T^*)$ is denoted by primal-dual Schnyder
wood. For more detailed background see~\cite{f-lspg-04} 
and~\cite{fz-swaos-08}.

Let a 3-connected plane graph $G$ and a primal-dual Schnyder wood for $G$
be given. Following the approach of Gon\c{c}alves, L{\'e}v{\^e}que and Pinlou
we first construct an auxiliary graph $H$. The SLTR of $H$ will be the
dissection of a triangle which is the primal-dual contact representation of
$G$.  In contrast to~\cite{GoncalvesLP12} we work with an FAA on $H$
and not with a contact family of pseudosegments.

The vertices of $H$ are the edges of $G$ including the half edges at
the suspensions. The vertices corresponding to the half edges are the
suspensions of $H$. The edges of $H$ correspond to the angles of $G$,
i.e., if $e$ and $e'$ are both incident to a common vertex $v$ and a
common face $f$, then $(e,e')$ is an edge of $H$. The faces of $H$ are
in bijection to vertices and faces (dual vertices) of $G$. 
In the context of knot theory this graph $H$ is
known as the {\em medial graph} of $G$. 

The graph $H$ inherits a plane drawing from $G$. The faces of $H$ are
in bijection to the vertices and faces of $G$. In an SLTR of $H$ we
need three corners in every face, moreover, every vertex of $H$
(except the three suspensions) has to be the corner for three of its
four incident faces. A corner assignment with these two properties is
obtained form the orthogonal arcs of the surface, i.e., if $s$ is a
vertex and $g$ is a face of $H$, then $s$ is one of the three
designated corners for $g$ if and only if in $g$ there is an orthogonal arc
ending in $s$. The corner assignment is equivalent to an FAA, an angle
of $s$ is to be flat if the two edges of $H$ forming the angle belong
to the same flat of the orthogonal surface.  An example is shown in
Figure~\ref{fig:Handoutl}.

\begin{figure}[htb]
                \begin{center}
		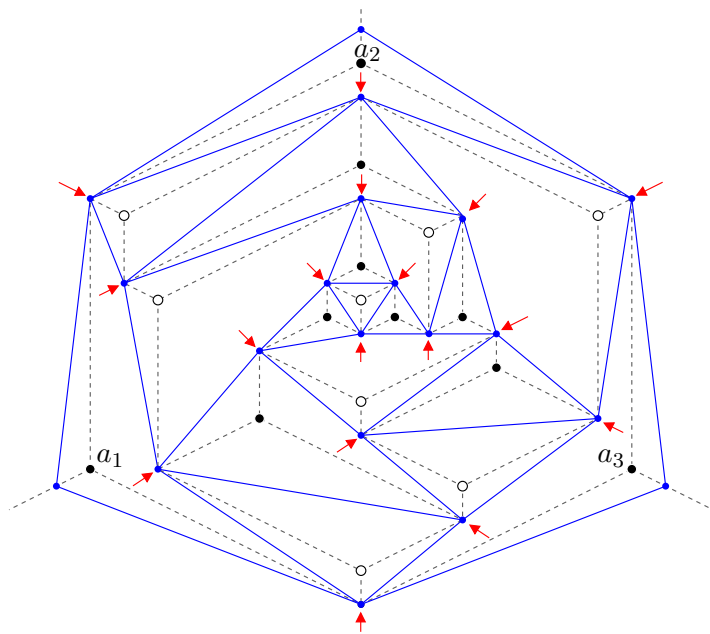
		\end{center}
          \caption{The graph $H$ (in blue) is drawn on top of an orthogonal
          surface (in dashed grey). The flat angles of an FAA are given by the
          red arrows. }\label{fig:Handoutl}
\end{figure}

The family of pseudosegments corresponding to this FAA is precisely
the family defined by Gon\c{c}alves, L{\'e}v{\^e}que and Pinlou.  This
family of pseudosegments also has a nice description in terms of
the flats. In fact there is a bijection between the pseudosegments and
bounded flats. A flat $F$ whose boundary consists of $2k$ orthogonal
arcs, contains $k$ saddle-points of the surface, these are the vertices
of $H$ on $F$. These vertices induce a path $P_F$ in $H$. Every
internal vertex of $P_F$ has a flat angle in $F$ and is, hence, assigned,
see Figure~\ref{fig:flat}. If $F$ is a flat which is constant in
coordinate $i$, then within $P_F$ one of the endpoints is maximal in
coordinate $i-1$ and the other is maximal in coordinate $i+1$. We call
them the {\em left-end} and the {\em right-end} of $P_F$,
respectively. In each of the three unbounded flats we have two
suspensions of $H$ as end-vertices for the path.

A flat is called {\em
rigid} if $P_F$ is a monotone path with respect to  coordinates
$i-1$ and $i+1$. The flat $F$ shown in the left part of
Figure~\ref{fig:flat} is not-rigid, the path $P_F$ is not monotone
with respect to coordinate $i+1$. An orthogonal surface is rigid if 
all its bounded flats are rigid. It has been shown in
\cite{f-gepg-03} and \cite{fz-swaos-08} that every Schnyder wood has a 
geodesic embedding on some rigid orthogonal surface. From now on
we assume that the given orthogonal surface is rigid, this assumption
will be critical in the proof of Proposition~\ref{lem:SchnyderGFAA}.

\begin{figure}[htb]
                 \begin{center}
		\includegraphics[scale=.4]{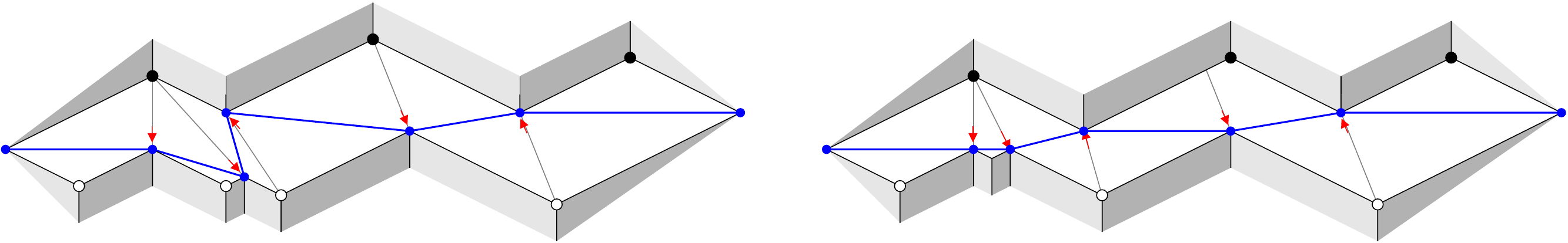}
		\end{center}
                \caption{Two combinatorially equivalent sketches of a
                  typical flat. The left one is non-rigid, the right
                  one is rigid.  The gray edges belong to the primal
                  dual Schnyder wood, they prescribe the edges of $H$
                  on the flat. The sequence of $H$ edges is a
                  pseudosegment of the FAA.  }\label{fig:flat}
\end{figure}

To prove that the FAA thus defined is a good FAA, we use the structure
of the flats. First we note that the flats are naturally partitioned
into three classes, let $\FF_i$ be the set of flats of color $i$, i.e,
of the flats whose boundary consists of orthogonal arcs in directions
$i-1$ and $i+1$. 

\begin{proposition}\label{lem:SchnyderGFAA}
The flat angle assignment in $H$ as defined above is a Good FAA.
\end{proposition}

\begin{proof}
It is enough to show that every simple outline cycle has at least
three combinatorially convex corners
(Lemma~\ref{lem:simpleCycles}). Let $\gamma$ a simple outline cycle
in $H$. We consider $\gamma$ with its embedding into the rigid orthogonal
surface. 

On $\gamma$ we specify some special combinatorially convex vertices, they will
be called {\em candidates}. The candidates are not
necessarily distinct but we can show that at least three of them are
pairwise distinct. This is sufficient to prove the proposition.

The candidates come with a color. We now describe how to identify the
candidates of color $i$.  If $\gamma$ contains the suspension of color
$i$, then by (K1) this is a combinatorially convex vertex for $\gamma$
and we take it as the candidate.  Otherwise, consider the flat $F$
that has the maximal $i$ coordinate among all flats in $\FF_i$ that
contain a vertex from $\gamma$. Let $I$ be a path in $\gamma \cap F$.
As candidates of color $i$, we take the the endpoints of $I$.
Of course, if $I$ consists of just one vertex we only have one
candidate.

\noindent
\textit{Claim.}
The candidates are combinatorially convex.

A primal-saddle of $F$ is a corner between two vertices of $G$
and a dual-saddle is a corner between two dual vertices.
The vertices of $H$ in $F$ come in four types, left-end, right-end, 
primal-saddle and dual-saddle. 

A primal-saddle of $F$ has two edges in $H$,
that reach to a flat in $\FF_i$ with $i$ coordinate larger than $F$.
From the choice of $F_i$, we know that these two edges do not belong to $\gamma$.
Therefore, with a primal-saddle in $I$, both neighbors in $P_F$ also
belong to $\gamma$ and hence to $I$. Therefore, a primal-saddle is not an end of $I$ and thus not a candidate. 

If an end $z$ of $I$, is a dual-saddle, then it has an edge~$e$ of
$P_F$ that does not belong to $\int(\gamma)$. The edge~$e$ is part of
the angle at $z$ that belongs to the face to which $z$ is assigned,
i.e., $z$ is assigned to a face outside of $\gamma$. This shows that
$z$ is combinatorially convex by (K3).

If $z$ is an end of $P_F$. Consider the flat $F'$ that contains
two $H$-edges incident to $z$. The rigidity of $F'$ implies that 
$P_{F'}$ contains an edge $e$ incident to $z$ that reaches to a
flat in $\FF_i$ with $i$ coordinate larger than $F$. Hence, 
edge $e$  does not belong to $\gamma$ and not to $\int(\gamma)$.
The edge $e$ is part of
the angle at $z$ that belongs to the face to which $z$ is assigned.
Again $z$ is combinatorially convex by (K3).

This concludes the proof of the claim.
\qedclaim

It can happen that a candidate $z_i$ of color $i$ and a candidate $z_j$
of color $j$ coincide. We have to show that in total we obtain at least
three different candidates.  

Let candidate $z_i$ be a dual-saddle at a flat $F_i$ of color $i$. 
Let $G_{i-1}$ and $G_{i+1}$ be the other two flats incident to $z_i$. The
two edges of $H$ in $G_{i-1}$ and $G_{i+1}$ belong to $\int(\gamma)$
and show that $G_{i-1}$ and $G_{i+1}$ are not maximal in their
respective colors. Hence, $z_i$ is a candidate only in color $i$.

It remains to look at the left-ends and right-ends of paths $P_F$.
Let $z$ be the endpoint of paths $P_{F_i}$ and $P_{F_j}$.
We claim that the two edges in $F_i$ and $F_j$ incident to 
$z$ belong to $\gamma$. Otherwise, consider an edge $e$ of $\gamma$
on the third flat $G$ incident to $z$. This edge either reaches a flat
of color $i$ higher than $F_i$ in coordinate $i$ or a flat
of color $j$ higher than $F_j$ in coordinate $j$,
This  contradicts the maximality of 
either $F_i$ or $F_j$. Since $z$ is incident to edges in $F_i$ and
$F_j$ we know that it is not the only candidate of color $i$ and not
the only candidate of color $j$. 

This is enough to show that there are at least three pairwise
different candidates.
\end{proof}

As every 3-connected plane graph $G$ has a Schnyder wood, we can define the
auxiliary graph $H$ and an FAA of $H$ can be obtained as
described. Proposition~\ref{lem:SchnyderGFAA} shows that this FAA is good. 
We have thus reproved the theorem:

\begin{theorem}
  Every 3-connected plane graph admits a primal-dual triangle contact
  representation.
\end{theorem}

In the proof we have worked with the skeleton graph $H$ of the primal-dual triangle representation.  We continue by asking which graphs $H$ can
serve as skeleton graphs for a primal-dual representation of some graph.

If a dissection of a triangle is a primal-dual triangle contact representation
of some graph, then there is a 2 coloring of the triangles. Hence, the
skeleton graph $H$ is Eulerian, i.e., all the vertex degrees are even.  It is
also evident that only degrees 4 and 2 are possible.  

\begin{definition}[Almost 4-regular]
  A plane graph is almost 4-regular\footnote{Almost 4-regular graphs
    are Laman graphs. The number of edges is twice the number of vertices
    minus three and this is an upper bound for each subset of the vertices.}
  if:
\Bitem There are three vertices of degree 2 on the outer face,
\Bitem All the other vertices have degree 4.
\end{definition}

With the following theorem, we show that deciding whether an almost 4-regular plane graph has an
SLTR is equivalent to deciding whether the underlying graph is 3-connected.

\begin{theorem}
  An almost 4-regular plane graph $H$ has an SLTR if and only if it is the
  medial graph of an interiorly 3-connected graph, or $H=C_3$.
\end{theorem}
\begin{proof}
  Let $H\neq C_3$ be an almost 4-regular plane graph and let $R$ be a SLTR of
  $H$. The three suspensions in $R$ are the three degree two vertices. Since
  $H$ is even, the dual is a bipartite graph. 
We abuse
  notation and denote the bounded faces in $R$ that contain the suspension
  vertices, with \emph{suspension of the dual}. Since they are all
  adjacent to the outer face of $R$, the suspensions are in the same
  color class of the bipartition, say in the white class.
\begin{figure}[ht]
\begin{center}
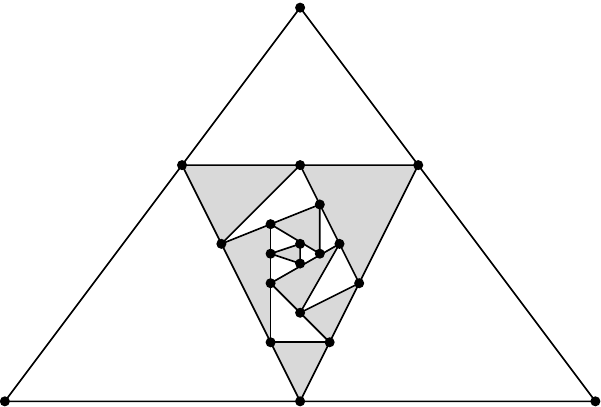
\end{center}
\caption{An SLTR.}\label{fig:SLTRColor}
\end{figure}

Let $G^{\triangle}$ be the graph whose vertices correspond to the
white triangles of $R$ together with an extra vertex $v_\infty$.
The edges of $G^{\triangle}$ are the contacts between white triangles
together with an edge between each of the suspensions and $v_\infty$.
The degree of $v_\infty$ is three and each corner of a white triangle
is responsible for a contact, hence, every vertex of $G^{\triangle}$
has degree at least three.

\emph{Claim.}
${G}^{\triangle}$ is 3-connected.

Suppose there is a separating set $U$ of size at most 2.  Let $C$ be
be a component of $G^{\triangle}\backslash U$ such that
$v_\infty\not\in C$. The convex hull $H_C$ of the corners of triangles
in $C$ has at least 3 corners. Covering all the corners of $H_C$ with
only two triangles results in a corner $p$ of $H_C$ that has a contact to
a triangle $T\in U$ such that $p$ has an angle larger than $\pi$ 
in the skeleton of $C+T$. Since $p$ is a vertex of $H$ and angles  
larger than $\pi$ do not occur at vertices of degree 4 of an SLTR,
this is a contradiction.
\qedclaim

By construction $H$ is just the medial graph of $G^{\triangle}$.
\end{proof}

\section{Conclusion and Open Problems}\label{conclusion}

We have given necessary and sufficient conditions for a 3-connected
planar graph to have an SLT Representation. Given an FAA and a set of
rational parameters~$\{\lambda_i\}_i$, the solution of the harmonic
system can be computed in polynomial time. Checking whether a solution
is degenerate can also be done in polynomial time.  Hence, we
can decide in polynomial time whether a given FAA corresponds
to an SLTR. In other words, checking whether a given
FAA is a GFAA can be done in polynomial
time.  However, most graphs admit different FAAs of which only some
are good. We are not aware of an effective way of finding a GFAA. 
Therefore, we have to leave this problem open: 
Is the recognition of graphs that have an SLTR (GFAA) in $P$?



\smallskip

Given a 3-connected planar graph and a GFAA, interesting optimization
problems arise, e.g. find the set of parameters $\{\lambda_i\}_i$ such
that the smallest angle in the graph is maximized, or the set of
parameters such that the length of the shortest edge is maximized.

\smallskip

Gon\c{c}alves, L{\'e}v{\^e}que and Pinlou conjectured that every 3-connected planar graph admits a primal-dual contact representation by right triangles, where all triangles have a horizontal and a vertical side and the right angle is bottom-left for primal vertices and top-right otherwise~\cite{GoncalvesLP12}. To the best of our knowledge this is still open. Perhaps the new proof could give more insight into this problem. 



\small 
\bibliographystyle{spmpsci}
\bibliography{bibSLTR} 



\end{document}